\documentclass[a4paper,UKenglish,cleveref, autoref, thm-restate]{lipics-v2021}

\DeclareMathOperator*{\argmax}{arg\,max} 

\title{Improved Approximations for Extremal Eigenvalues of Sparse Hamiltonians} 

\titlerunning{Approximating Sparse Hamiltonians} 

\author{Daniel Hothem}{Quantum Algorithms and Applications Collaboratory, Sandial National Laboratories, Livermore, CA 94550, USA}{dhothem@sandia.gov}{https://orcid.org/0000-0002-5628-9945}{}

\author{Ojas Parekh}{Quantum Algorithms and Applications Collaboratory, Sandial National Laboratories, Albuquerque, NM 87123, USA}{odparek@sandia.gov}{https://orcid.org/0000-0003-2689-9264}{}

\author{Kevin Thompson}{Quantum Algorithms and Applications Collaboratory, Sandial National Laboratories, Albuquerque, NM 87123, USA}{kevthom@sandia.gov}{}{}

\authorrunning{D. Hothem, O. Parekh, and K. Thompson} 

\Copyright{Daniel Hothem, Ojas Parekh, and Kevin Thompson} 



\category{} 

\relatedversion{}
\relatedversiondetails{Previous Version}{https://arxiv.org/abs/2301.04627}

\funding{This material is based upon work supported by the U.S. Department of Energy, Office of Science, Office of Advanced Scientific Computing Research, National Quantum Information Science Research Centers, Exploratory Research for Extreme Scale Science program.  Support is also acknowledged from the Accelerated Research in Quantum Computing program under the same office.}

\acknowledgements{We thank Yaroslav Herasymenko for an insightful contribution to \Cref{lem:gaussianmixture}. This article has been authored by an employee of National Technology \& Engineering Solutions of Sandia, LLC under Contract No. DE-NA0003525 with the U.S. Department of Energy (DOE). The employee owns all right, title and interest in and to the article and is solely responsible for its contents. The United States Government retains and the publisher, by accepting the article for publication, acknowledges that the United States Government retains a non-exclusive, paid-up, irrevocable, world-wide license to publish or reproduce the published form of this article or allow others to do so, for United States Government purposes. The DOE will provide public access to these results of federally sponsored research in accordance with the DOE Public Access Plan \url{https://www.energy.gov/downloads/doe-public-access-plan}.
}

\nolinenumbers 

\EventEditors{Omar Fawzi and Michael Walter}
\EventNoEds{2}
\EventLongTitle{18th Conference on the Theory of Quantum Computation, Communication and Cryptography (TQC 2023)}
\EventShortTitle{TQC 2023}
\EventAcronym{TQC}
\EventYear{2023}
\EventDate{July 24--28, 2023}
\EventLocation{Aveiro, Portugal}
\EventLogo{}
\SeriesVolume{266}
\ArticleNo{6}

\begin{document}

\hideLIPIcs

\maketitle

\begin{abstract}
    We give a classical $1/(qk+1)$-approximation for the maximum eigenvalue of a $k$-sparse fermionic Hamiltonian with strictly $q$-local terms, as well as a $1/(4k+1)$-approximation when the Hamiltonian has both $2$-local and $4$-local terms.  More generally we obtain a $1/O(qk^2)$-approximation for $k$-sparse fermionic Hamiltonians with terms of locality at most $q$. Our techniques also yield  analogous approximations for $k$-sparse, $q$-local qubit Hamiltonians with small hidden constants and improved dependence on $q$.
\end{abstract}

\section{Introduction}
Finding the ground state energy of systems of particles is a fundamental problem of quantum mechanics. Finding the ground state energies of local Hamiltonians is believed to be difficult for both classical and quantum computers \cite{K06, K02}. Instead, it is often easier to find classical and quantum approximations to these ground state energies. In this paper, we consider approximations to the extremal eigenvalues of a local, $k$-sparse fermionic Hamiltonian: 
\begin{equation*}
    H=\sum_{\Gamma} H_\Gamma c^{\Gamma}.
\end{equation*}
Here $H$ is a fermionic Hamiltonian with real coefficients $H_{\Gamma}$, where ignoring phase factors, each term $c^{\Gamma}$ is a product of $q$ Majorana operators (i.e.,\ $H$ is $q$-local with $q$ even) and each Majorana operator appears in at most $k$ non-zero terms (i.e.,\ $H$ is $k$-sparse).  We let $m = \sum_{\Gamma} |H_\Gamma|$.

Our main technical contribution is a carefully designed graph $G$, whose vertices correspond to the terms in $H$. We are able to construct states that achieve better approximations than in previous works by finding a suitably large independent set in $G$. This work is similar to (but distinct from) recent work by Herasymenko, Stroeks, Helsen, and Terhal \cite{H22} in which a similar graph is used to find \emph{diffuse sets} of Majorana monomials from which they construct a state. Herasymenko et al.\ also work with a graph whose vertices correspond to the terms in $H$; however, their edge set is different. Our new edge set also allows us to generalize our results beyond the $q = 4, 2$ case handled by Herasymenko et al., and to prove better approximation ratios.  

\begin{table}[t]
    \centering
    \caption{Main results contrasted with the previous state of the art.}
    \begin{tabular}{|c|c|c|c|}
    \hline
    \multicolumn{2}{|c|}{Hamiltonian} & Our result & Previous result \\
    \hline
    \multirow{3}{*}{fermionic} & $k$-sparse, strictly $q$-local fermionic & $1/(qk+1)$ & $1/\mathcal{O}(q^2k^2)$\cite{H22} \\
    & $k$-sparse, $4,2$-local & $1/(4k+1)$ & $1/\mathcal{O}(k^2)$\cite{H22}\\
    & $k$-sparse, $q$-local & $1/\mathcal{O}(qk^2)$ & N/A \\
    \hline
    \multirow{3}{*}{qubit} & $k$-sparse, strictly $q$-local & $1/(qk+1)$ & $3^{-q/2}/(4qk)$\cite{HM17} \\
    & $k$-sparse, $2$-local & $1/(2k+1)$ & 1/(24k)\cite{HM17}\\
    & $k$-sparse, $q$-local & $1/\mathcal{O}(qk^2)$ & $3^{-q/2}/(4qk)$\cite{HM17} \\
    \hline
    \end{tabular}
    \label{tab:results}
\end{table}

\Cref{tab:results} summarizes our main results. We also list the previously known best results. The table is split into two sections: (1) fermionic Hamiltonians and (2) qubit Hamiltonians. Although our work does not focus on qubit Hamiltonians, our proof ideas furnish results that improve upon the previously known best results (see Sections~\ref{sec:contributions} and \ref{sec:extensions}).

\section{Contextualizing our results}\label{sec:contributions}
Bravyi, Gosset, Koenig, and Temme~\cite{Bra18} were the first to suggest approximation algorithms for the largest eigenvalue of fermionic Hamiltonians using fermionic Gaussian states, achieving a $1/\mathcal{O}(n\log(n))$-approximation ratio for generic $4$-local fermionic Hamiltonians.  They also asked whether Gaussian states might provide a constant-factor approximation. Among other results, Hastings and O'Donnell~\cite{HO22} subsequently demonstrated that Gaussian states offer at best a $1/\Omega(\sqrt{n})$-approximation for a class of $4$-local fermionic Hamiltonians, known as the Sachdev-Ye-Kitaev (SYK) model.  Hamiltonians in the SYK model are dense $4$-local Hamiltonians, hence the work of Hastings and O'Donnell left open the possibility of a constant-factor Gaussian approximation algorithm for models with sparse Hamiltonians. Sparse Hamiltonians are a natural class of Hamiltonians to study. Examples, such as the Fermi-Hubbard model, are ubiquitous~\cite{Hub63}. 

Recent work by Herasymenko, Stroeks, Helsen, and Terhal~\cite{H22} proves the existence of such constant-factor approximations. They show that $\lambda_{\max}(H) \geq m/Q$, where $\lambda_{\max}(H)$ is the largest eigenvalue of $H$ and $Q=q(q-1)(k-1)^2+q(k-1)+2$ for a $k$-sparse strictly $q$-local fermionic Hamiltonian. Herasymenko et al.\ also prove an improved ratio of $Q = 12(k+1)^2+4(k-1)+2$ when specializing to $k$-sparse, $4,2$-local fermionic Hamiltonians. Our work directly improves upon these results (see \Cref{tab:results}). Our work also removes the conditions on system size present in Herasymenko et al. This leads to immediate improvements in Herasymenko et al.'s work on the sparse SYK model. All of the above results are obtained by efficient classical algorithms producing descriptions of Gaussian states. We refer the reader to \cite{H22} for further background, motivation, and applications to the SYK model. Finally, Herasymenko et al.'s result do not extend to $k$-sparse, $q$-local fermionic Hamiltonians (i.e., where all terms have locality \emph{at most} $q$). To our knowledge, our $1/\mathcal{O}(qk^2)$ approximation is the first of this kind.  It remains an open question whether this may be improved to $1/\mathcal{O}(qk)$. 

Results of the above flavor were obtained for traceless $k$-sparse qubit Hamiltonians with constant locality by Harrow and Montanaro~\cite{HM17}, who show that $\lambda_{\max}(H) \geq \Omega(m/k)$ using product states, where $k$-sparse and $m$ are defined analogously as above; bounds upon which our ideas give a constant-factor improvement (see \Cref{tab:results}). They also give an improved bound with respect to the operator norm instead of the maximum eigenvalue: $\|H\| \geq \Omega(m/\sqrt{k})$.  In the fermionic case, we give a 2-local example with $\lambda_{\max}(H) = \|H\| = \Theta(m/k)$, showing that such an improvement is not possible (see \Cref{sec:optimality}) and that our result for the strictly $q$-local case is tight.  As noted in \Cref{tab:results}, our techniques also apply to the Hamiltonians considered by Harrow and Montanaro, yielding approximation guarantees with small hidden constants and improved dependence on $q$.

\section{Preliminaries}\label{sec:preliminaries}

In this section, we provide the necessary preliminaries for the rest of the work. We begin with an overview of fermionic Hamiltonians before providing the necessary background on Gaussian states. This section draws upon \cite{H22, HO22, Bra05}. 

\subsection{Fermionic Hamiltonians}\label{ssec:fermionic_hamiltonians}

Fermionic Hamiltonians describe systems of fermionic particles, such as electrons. For our purposes, it is easiest to express a fermionic Hamiltonian in terms of Majorana operators. Throughout, we use the notation $[n]:=\lbrace 1,\dots, n\rbrace$. We also use the notation $\mathcal{E} = \{\Gamma \subseteq [n] \mid H_\Gamma \neq 0\}$ to denote the set of non-zero terms in a Hamiltonian.

\begin{definition}\label{def:majorana_operators}
    Given $n$ fermionic modes, a set of $2n$ traceless and Hermitian operators $\lbrace c_i\rbrace_{i=1}^{2n}$ are Majorana operators if they satisfy $c_ic_j + c_jc_i = 2\delta_{ij}$ for all $i, j \in [2n]$.
\end{definition}

\begin{definition}\label{def:fermionic_hamiltonian}
    Let $\lbrace c_i\rbrace$ be a collection of $2n$ Majorana operators endowed with an ordering (say the lexicographic ordering). A fermionic Hamiltonian has the form:
    \begin{equation}\label{eq:Hamiltonian}
    H=\sum_{\Gamma\subseteq [2n]} H_\Gamma c^{\Gamma},
\end{equation}
where $\Gamma \subseteq [2n]$ has even order, $c^{\Gamma}$ is the product of Majorana operators appearing in $\Gamma$ (ordered lexicographically), and the $H_{\Gamma} \in \mathbb{R}$. Note that $c^{\Gamma}$ may contain an additional pre-factor of $i$ in order to satisfy hermiticity (e.g., when $\vert\Gamma\vert = 2$).
\end{definition}

\begin{definition}\label{def:locality}
    If $H$ is a fermionic Hamiltonian defined in terms of $2n$ Majorana operators, then $H$ is $q$-local if there exists $q \in \mathbb{N}$ such that each non-zero term in $H$ has locality at most $q$, that is, for all $\Gamma \subseteq [2n]$ with $H_{\Gamma}\neq 0$, $\vert \Gamma\vert\leq q$. $H$ is strictly $q$-local if $\vert\Gamma\vert = q$ for all non-zero summands.
\end{definition}

\begin{definition}\label{def:sparse}
    Let $H$ be a fermionic Hamiltonian on $2n$ Majorana operators. Then $H$ is $k$-sparse if each Majorana operator $c_i$ appears in at most $k$ non-zero terms, that is, for all $i \in [2n]$, $|\{ \Gamma \in \mathcal{E} \mid i \in \Gamma \}|\leq k$.
\end{definition}

\subsection{Gaussian states}\label{ssec:gaussian_states}
First note that for any real, orthogonal matrix $R \in O(2n)$, the transformation
\begin{equation}\label{eqn:rotation}
    \tilde{c}_i = \sum_{j=1}^{2n}{R_{ij}c_j},
\end{equation}
gives rise to a new set of $[2n]$ Majorana operators $\lbrace \tilde{c}_i\rbrace$.

\begin{definition}\label{def:gaussian_state}
Let $\lbrace c_{i}\rbrace$ be a set of $[2n]$ Majorana operators, $R \in O(2n)$, and $\lbrace \tilde{c}_i\rbrace$ defined as in \ref{eqn:rotation}. For any assignment $\lambda_1, \dots, \lambda_{n} \in [-1,1]$, the following state is a (mixed) fermionic Gaussian state:
\begin{equation}
    \rho = \frac{1}{2^n}\prod_{j=1}^{n}{\Big(\mathbb{I} + i\lambda_j\tilde{c}_{2j-1}\tilde{c}_{2j}\Big)}
\end{equation}
The state $\rho$ is pure when $\lambda_j \in \lbrace \pm 1\rbrace$ for all $j \in [n]$.
\end{definition}

Fermionic Gaussian states exhibit several nice properties. Not only are they the ground states of homogeneous 2-local fermionic Hamiltonians \cite{Bra05}, but their higher-order correlates are efficiently computable from their \emph{correlation matrix}. If $\rho$ is defined as in \Cref{def:gaussian_state}, then the correlation matrix $M$ of $\rho$ is the real, antisymmetric $2n \times 2n$ matrix with entries defined as:
\begin{equation}\label{def:correlation_matrix}
        M_{ij} = \frac{i}{2}Tr(\rho[c_i, c_j]).
    \end{equation}
The higher-order correlates of $\rho$ can be computed via Wick's formalism:
\begin{equation}
    Tr(c^{\Gamma}\rho) = Pf(M_{\Gamma}),
\end{equation}
where $M_{\Gamma}$ is the $\vert\Gamma\vert\times\vert\Gamma\vert$ submatrix of $M$ containing only the ordered rows and columns in $\Gamma$ and $Pf(\cdot)$ is the matrix Pfaffian. Finally, for any Gaussian state $\rho$, the set $\lbrace \lambda_j\rbrace$ and $M$ are connected by the following lemma:

\begin{lemma}[Bra05]\label{lem:gaussian_rotation}
    For any Gaussian state $\rho$ with correlation matrix $M$, there exists some $R \in O(2n)$ such that the adjoint action of $O(2n)$ on $M$ block-diagonalizes $M$ into the following form:
    \begin{equation}
        M = R\bigoplus_{j=1}^{n}{\begin{pmatrix}
                                0 & \lambda_{j} \\
                                -\lambda_j & 0  \end{pmatrix}}R^T,
    \end{equation}
    where the $\lambda_j$ are the same as in \Cref{def:gaussian_state}. Thus, every real, anti-symmetric matrix $M$ is the correlation matrix for some Gaussian state $\rho$.
\end{lemma}

\section{Main approximation algorithm}\label{sec:main_result}

In this section, we demonstrate our main technical ideas by proving an approximation ratio for $k$-sparse Hamiltonians with both $4$-local and $2$-local terms. We chose this specific case as it highlights all of our technical ideas, while also being the most physically interesting case. In \Cref{sec:extensions} we show how these ideas generalize to the other cases described in \Cref{tab:results}. 

\begin{theorem}\label{thm:main}
There is a classical polynomial time algorithm that, given as input the weights $\{H_\Gamma\}$ of some $k$-sparse and $4,2$-local fermionic Hamiltonian $H$, returns a description of a quantum state $\rho$ achieving energy 
\begin{equation*}
Tr(H \rho) \geq \frac{1}{4k+1} \sum_{\Gamma} |H_\Gamma| \geq \frac{1}{4 k +1}\lambda_{max}(H).
\end{equation*}
\end{theorem}
\begin{proof}
Define a graph $G=(V, E)$ with vertices corresponding to the nonzero terms in the Hamiltonian (i.e.,\ $V = \mathcal{E}$).  The graph $G$ may contain vertices corresponding to $2$-local or $4$-local terms. We include an edge $(v_{\Gamma}, v_{\Gamma'}) \in E$ if and only if one of the following conditions is met: 
\begin{enumerate}
    \item[(i)]\label{cond:one} $c^\Gamma$ and $c^{\Gamma'}$ share one or more Majorana operators (i.e.,\ $\Gamma \cap \Gamma' \neq \emptyset$), or
    \item[(ii)]\label{item:2} $\Gamma$ and $\Gamma'$ are disjoint and $\Gamma \cup \Gamma' \in \mathcal{E}$. 
\end{enumerate}
If there are $m$ nonzero terms in the Hamiltonian then the graph $G$ has $m$ vertices, and the degree of a vertex in the graph is at most $4k$.  We can see the latter as follows.  Fix some vertex $v_\Gamma$. By construction, 
\begin{equation}\label{eq:deg}
    \text{deg}(v_{\Gamma}) = \vert \lbrace (\Gamma, \Gamma') \in \mathcal{E}\times \mathcal{E} \mid \text{$\Gamma$ and $\Gamma'$ satisfy (i) or (ii)}\rbrace \vert.
\end{equation}

We consider two cases:
\begin{itemize}
    \item[-] $\Gamma$ is 4-local. Consider an edge $(v_{\Gamma}, v_{\Gamma'})$. As $H$ contains no 6-local or 8-local terms, $\Gamma \cap \Gamma' \neq \emptyset$. As $H$ is $k$ sparse, there are at most $4k$ $\Gamma'$ for which this can occur.
    
    \item[-] $\Gamma$ is 2-local. Let $a$ equal the number of $4$-local Hamiltonian terms overlapping with $\Gamma$, and let $b$ equal the number of $2$-local terms overlapping with $\Gamma$.  We claim that the degree of $v_\Gamma $ is at most $2a+b$.

    There are $b$ 2-local $\Gamma'$ satisfying (i) with $\Gamma$.  Each $2$-local $\Gamma'$ satisfying (ii) results in a unique 4-local $\Gamma \cup \Gamma' \in \mathcal{E}$ overlapping with $\Gamma$, hence there at most $a$ such $\Gamma'$.  Finally, no $4$-local $\Gamma'$ may satisfy (ii), and there are $a$ 4-local $\Gamma'$ satisfying (i).    
    
    Since $\Gamma$ overlaps with at most $2k$ $\Gamma'$, we have $a+b \leq 2k$ so that $2a+b \leq 4k$.
\end{itemize}

By Brooks' Theorem we can in polynomial time find a coloring of the vertices of $G$ with at most $4k+1$ colors.  This means we can partition the vertices into at most $4k+1$ independent sets, $\{S_1, ..., S_t\}$, with one of these sets having at least a $1/(4k+1)$ fraction of the sum of the absolute values of the weights:
\begin{equation}\label{eqn1}
\sum_{\Gamma} |H_\Gamma|=\sum_{S_i} \sum_{\Gamma\in S_i} |H_{\Gamma}| \leq (4k+1) \max_i \sum_{\Gamma \in S_i} |H_{\Gamma}|.
\end{equation}
It follows from \Cref{eqn1} that
\begin{equation*}
    \max_i \sum_{\Gamma \in S_i} |H_{\Gamma}| \geq \frac{1}{(4k+1)}\sum_{\Gamma}|H_\Gamma|.
\end{equation*}

Define $S_j = \argmax_j \sum_{\Gamma \in S_j}{\vert H_{\Gamma}\vert}$, and consider the following state:
\begin{equation}\label{eq:state}
\rho=\frac{1}{2^n} \prod_{\Gamma \in S_j} (\mathbb{I}+\text{sign}(H_\Gamma)c^\Gamma).
\end{equation}

We claim that $\rho$ is a valid quantum state and obtains objective $\sum_{\Gamma \in S_j} |H_\Gamma|$.  By definition, $\rho$ is proportional to a projector on a stabilizer state with stabilizer generators given by $c^\Gamma$ for $\Gamma \in S_j$: Observe that $[c^\Gamma, c^{\Gamma'}]=0$ for all $\Gamma, \Gamma'\in S_j$ since $S_j$ is an independent set.  Hence, $\rho$ is the product of commuting projectors and must be positive semidefinite. 

To see that $\rho$ obtains the desired objective, we first expand the product in \Cref{eq:state} as a sum and consider products of two or more terms, $\sigma = \prod_p c^{\Gamma_p}$ for $\Gamma_p \in S_j$. If any of the $\Gamma_p$ are 4-local or $p \geq 3$, $\sigma$ cannot be proportional to a term of $H$ since the $\Gamma \in S_j$ are disjoint, and no cancellation in products of Majorona operators can occur.  The remaining case is a product of two 2-local operators.  For any such $\Gamma,\Gamma' \in S_j$, by (ii) and because $S_j$ is an independent set, the product $c^\Gamma c^{\Gamma'}$ cannot be proportional to $c^{\Gamma''}$ for any $\Gamma'' \in \mathcal{E}$.

Hence we have
\begin{align*}
    Tr(\mathbb{I} \rho) &= 1,\\
    Tr(c^\Gamma \rho) &= \text{sign}(H_\Gamma) \quad \forall\Gamma \in S_j,\text{ and }\\
    Tr(c^\Gamma \rho) &= 0\quad \forall\Gamma \in \mathcal{E}\setminus S_j.
\end{align*}
This yields the desired claim that $\rho$ is a normalized state for which
\begin{equation*}
Tr(H \rho)=\sum_\Gamma H_\Gamma Tr(c^\Gamma \rho)=\sum_{\Gamma \in S_j} H_\Gamma Tr(c^\Gamma \rho) =\sum_{\Gamma \in S_j} |H_\Gamma| \geq \frac{1}{4k+1} \sum_\Gamma |H_\Gamma|.
\end{equation*}
\end{proof}

\section{Conversion to a Gaussian state}
The $\rho$ constructed in \Cref{thm:main} is, in fact, a mixture of Gaussian states. This is proven in the following lemma. This implies the existence of a Gaussian state with at least the same objective as $\rho$.

\begin{lemma}\label{lem:gaussianmixture}
The state $\rho$ defined in \Cref{eq:state} is a mixture of Gaussian states.
\end{lemma}
\begin{proof}
For each $\Gamma \in S_j$ let $M_{\Gamma}$ be the perfect matching of the operators in $\Gamma$ induced by the lexicographic ordering of $\Gamma$, and let $M$ be a perfect matching of the Majorana operators in $\{c_1, ...c_{2n}\}\setminus  \lbrace c_i \mid \exists \Gamma \in S_j \text{ with } i \in \Gamma\rbrace$ induced by the lexicographic ordering. Define the following Gaussian state:
\begin{equation}\label{eq:2}
    \rho'(z) = \frac{1}{2^n} \prod_{\Gamma \in S_j} \prod_{gh \in M_{\Gamma}} (\mathbb{I}+ z_{gh} \,\, i c_g c_h) \prod_{ rs \in M} (\mathbb{I} + z_{rs}\,\, i c_r c_s),
\end{equation}
where all $z_{gh}, z_{rs} \in \{\pm 1\}$. 

Consider the state $\rho''=\mathbb{E}_z [\rho'(z)]$ where for each $\Gamma$ the set $\{z_{gh}\}_{gh \in M_\Gamma}$ is uniformly random distributed over $\{\pm 1\}^{|M_\Gamma|}$ subject to the constraint:
\begin{equation}\label{eq:constraints}
\text{sign}\left[\left(\prod_{gh \in M_\Gamma} z_{gh}\,\, i c_g c_h \right) c^\Gamma \right]=\text{sign}(H_\Gamma) \quad \forall \Gamma\in S_j,
\end{equation}
where $\text{sign}(\pm \mathbb{I})$ is defined as $\pm 1$.  In other words, $\{z_{gh}\}_{gh \in M_\Gamma}$ is chosen as the uniform distribution over strings in $\{\pm 1\}^{|M_\Gamma|}$ which satisfy \Cref{eq:constraints}.  We will assume further that $\{z_{gh}\}_{gh\in M_\Gamma}$ is independent of all other $\{z_{gh}\}_{gh\in M_{\Gamma'}}$ and that each $z_{rs}$ for $rs\in M$ is uniform and independent of all other random variables. 

We claim that $\rho = \rho''$. Begin by using independence to push the expectation past the first and third products in \Cref{eq:2}:
\begin{equation}
    \rho'' = \frac{1}{2^n}\prod_{\Gamma \in S_j}\Big(\mathbb{E}_{z}\Big[ \prod_{gh \in M_{\Gamma}} (\mathbb{I}+ z_{gh} \,\, ic_gc_h)\Big]\Big) \prod_{ rs \in M}\Big( \mathbb{E}_{z}\Big[(\mathbb{I} + z_{rs}\,\, ic_r c_s)\Big]\Big),
\end{equation}
We first focus on the final product. Observe that:
\begin{equation}
    \prod_{ rs \in M} \Big(\mathbb{E}_{z}\Big[(\mathbb{I} + z_{rs}\,\, ic_r c_s)\Big]\Big) = \mathbb{I}
\end{equation}
This follows from the independence of the $\lbrace z_{rs} \mid rs \in M\rbrace$ and because $\mathbb{E}_{z}[z_{rs}] = 0$ for all $rs \in M$. Hence:
\begin{equation}
    \rho'' = \frac{1}{2^n}\prod_{\Gamma \in S_j}\Big(\mathbb{E}_{z}\Big[ \prod_{gh \in M_{\Gamma}} (\mathbb{I}+ z_{gh} \,\, ic_gc_h)\Big]\Big).
\end{equation}
For fixed $\Gamma \in S_j$, we claim that:
\begin{equation}\label{eq:rhoequiv}
    \mathbb{E}_{z}\Big[ \prod_{gh \in M_{\Gamma}} (\mathbb{I}+ z_{gh} \,\, ic_gc_h)\Big] = \mathbb{I}+\text{sign}(H_{\Gamma})c^{\Gamma}.
\end{equation}
\Cref{lem:gaussianmixture} follows immediately from \Cref{eq:rhoequiv}. For any strict subset $\Gamma' \subsetneq \Gamma$, define 
\begin{equation*}M_{\Gamma'\cap\Gamma} := \{gh \in M_{\Gamma}: g\in \Gamma', h\in \Gamma'\rbrace. 
\end{equation*}
We may then expand the left-hand side of \Cref{eq:rhoequiv} as:
\begin{align}\label{eq:rhoequiv_exp}
    \mathbb{E}_{z}\Big[ \prod_{gh \in M_{\Gamma}} (\mathbb{I}+ z_{gh} \,\, ic_gc_h)\Big] &= \mathbb{I} + \sum_{\Gamma'\subsetneq\Gamma}{\mathbb{E}_{z}\Big[\prod_{gh\in M_{\Gamma'\cap\Gamma}}{z_{gh}\,\, ic_g c_h}\Big]} + \mathbb{E}_{z}\Big[\prod_{gh\in M_{\Gamma}}{z_{gh}\,\, i c_g c_h}\Big] \\
    &= \mathbb{I} + \text{sign}(H_{\Gamma})c^{\Gamma}
\end{align}
The final expectation in \Cref{eq:rhoequiv_exp} evaluates to $\text{sign}(H_{\Gamma})c^{\Gamma}$ due to constraint \ref{eq:constraints}. The sum of expectations in \Cref{eq:rhoequiv_exp} disappears as the marginal distribution of the $z$ when restricted to a matching on a strict subset $\Gamma' \subsetneq \Gamma$ of size $|M_{\Gamma'\cap\Gamma}| = p$ is totally uniform over $\lbrace \pm 1\rbrace^{p}$. Therefore $\mathbb{E}_{z}[z_{gh}] = 0$ for any such matching.
\end{proof}

Although $\rho'(z)$ in \Cref{lem:gaussianmixture} is a Gaussian state for any $z$, the state $\rho''$ is a mixture of Gaussian states by definition.  However, we may  derandomize the choice of $z$ to obtain a Gaussian state.  We only require pairwise independence of the elements of $z$, hence using standard derandomization approaches, we can obtain a Gaussian state $\rho'(z)$ in polynomial time such that $Tr(H\rho'(z)) \geq Tr(H\rho'')$. 

\section{Extensions}\label{sec:extensions}
In this section, we demonstrate how our core approach in the proof of \Cref{thm:main} leads to improved classical approximation algorithms for the ground state energies of various sparse, local Hamiltonians. Each case is dealt with as its own corollary to \Cref{thm:main}.

\begin{corollary}[Strictly $q$-local Hamiltonians.]
Let $H$ be a $k$-sparse, strictly $q$-local fermionic Hamiltonian. There exists a classical polynomial time algorithm that, given $\lbrace H_{\Gamma}\rbrace$ as input, outputs a description of a quantum state $\rho$ achieving energy 
\begin{equation}
    Tr(H\rho) \geq \frac{1}{qk+1}\sum_{\Gamma}{\vert H_{\Gamma}\vert} \geq \frac{\lambda_{max}(H)}{qk+1}.
\end{equation}
\label{cor:strict_q}\end{corollary}
\begin{proof}
In this case we only need to include edges in $G$ between $v_\Gamma$ and $v_{\Gamma'}$ precisely when condition (i) holds, since (ii) is vacuous.  Consequently we may omit the second case below \Cref{eq:deg} and simply bound the degree as $qk$.  We then effectively replace ``4'' with $q$ in the remaining proof.
\end{proof}

\begin{corollary}[Hamiltonians with bounded locality.] Let $H$ be a $k$-sparse, $q$-local fermionic Hamiltonian. There exists a classical polynomial time algorithm that, given $\lbrace H_{\Gamma}\rbrace$ as input, outputs a description of a quantum state $\rho$ achieving energy 
\begin{equation}
    Tr(H\rho) \geq \frac{1}{Cqk^2}\sum_{\Gamma}{\vert H_{\Gamma}\vert} \geq \frac{\lambda_{max}(H)}{Cqk^2},
\end{equation}
for some constant $C \in \mathbb{R}$.
\label{cor:q_local} \end{corollary}
\begin{proof}
In this case we need an appropriate generalization of condition (ii) from \Cref{thm:main}.  Let us start by defining $G$ using only the condition (i); the maximum possible degree in $G$ is $qk$.  The purpose of (ii) in the proof is to ensure that for $\Gamma, \Gamma'$ in the independent set $S_j$, $c^{\Gamma}c^{\Gamma'}$ cannot be proportional to $c^{\Gamma''}$ for any $\Gamma'' \in \mathcal{E}$.  Note that if this happens, then $\Gamma''$ must contain both $\Gamma$ and $\Gamma'$. Thus it would suffice for our independent set $S_j$ in $G$ to satisfy the additional property that no $v_{\Gamma},v_{\Gamma'} \in S_j$ could have a common neighbor $v_{\Gamma''} \in V$ with $\Gamma,\Gamma' \subset \Gamma''$.  We could satisfy this by adding an edge in $G$ between all pairs $v_{\Gamma}$ and $v_{\Gamma'}$ with such a common neighbor.  By $k$-sparsity, the vertex $v_{\Gamma}$ has at most $k$ neighbors $v_{\Gamma''}$ in $G$ with $\Gamma \subset \Gamma''$.  Since any such $v_{\Gamma''}$ has degree at most $qk$, the degree of $v_{\Gamma}$ increases by at most $k(qk-1)$, and maximum degree in the resulting graph $G'$ is $O(qk^2)$.  Applying Brooks' Theorem in $G'$ produces the desired approximation.
\end{proof}

\begin{corollary}[Qubit Hamiltonians] Consider a $k$-sparse, $q$-local qubit Hamiltonian $H$ defined analogously to the fermionic Hamiltonian in \Cref{def:fermionic_hamiltonian}. Given the appropriate assumptions on the locality of $H$, there exists a classical polynomial time algorithm that, given as inputs the weights $\lbrace H_{\Gamma}\rbrace$, outputs a description of a quantum state $\rho$ achieving energy at least:

\begin{table}[h]
    \centering
    \begin{tabular}{|c|c|}
    \hline
    Hamiltonian & Energy \\
    \hline
    strictly $q$-local & $1/(qk+1)$ \\
    $k$-sparse, $2$-local & $1/(2k+1)$ \\
    $k$-sparse, $q$-local & $1/\mathcal{O}(qk^2)$ \\
    \hline
    \end{tabular}
\end{table}

\end{corollary} 
\begin{proof}
For qubit Hamiltonians, condition (i) in \Cref{thm:main} is modified to cover any pair of local terms which involve the same qubit, while condition (ii) is modified to be ``$\Gamma$ and $\Gamma'$ do not involve the same qubit.'' Our results for $k$-sparse and: (i) strictly $q$-local, (ii) $2$-local, and (iii) $q$-local qubit Hamiltonians follow from this modification and considering \Cref{cor:strict_q}, \Cref{thm:main}, and \Cref{cor:q_local} respectively. 
\end{proof}

\section{Optimality of our strictly q-local result}
\label{sec:optimality}

For $k$-sparse $H$ where all terms are $q$-local, since $\|H\| \geq \lambda_{\max}(H)$, our results show that 
\begin{equation*}
\|H\| \geq \lambda_{\max}(H) \geq \frac{m}{qk+1},
\end{equation*}
where we recall $m = \sum_{\Gamma}|H_{\Gamma}|$ and $\|\cdot\|$ denotes the operator norm.  We give an explicit family of fermionic 2-local $n$-sparse Hamiltonians $\{H_n\}_{n=1}^\infty$ demonstrating this bound is asymptotically tight (i.e., cannot be improved for all $q$ and $k$, up to constant factors).

Each $H_n$ is expressed as a sum of monomials in $2n$ Majorana operators $\{c_1, c_2, ..., c_{2n}\}$ satisfying the usual canonical anti-commutation relations.  For each $n$, partition $[2n]$ evenly into $A=\{1, ..., n\}$ and $B=\{n+1, ..., 2n\}$. Then:
\begin{equation*} 
H_n :=  \sum_{a\in A, b\in B} ic_a c_b= i \left(\sum_{a \in A} c_a\right) \left(\sum_{b\in B} c_b \right).
\end{equation*}

The eigenvalues of $H_n$ are easy to determine, define $R\in O(2n)$ as some orthogonal matrix satisfying:
\begin{equation*}
    R_{a, 1}=1/\sqrt{n}  \,\,\,\,\, \forall a \in A \text{      and     }  R_{b, 2}=1/\sqrt{n}\,\,\,\,\, \forall b \in B.
\end{equation*}
Note that this is well defined since the first two columns are orthonormal.  We can then define a new set of Majorana operators (also satisfying the canonical anti-commutation relations) by:
\begin{equation*}
\tilde{c}_i=\sum_{i=1}^{2n} R_{j, i} c_j.
\end{equation*}
In particular, we have
\begin{equation*}
\tilde{c_1}=\frac{1}{\sqrt{n}}\sum_{a \in A} c_a \text{     and    } \tilde{c_2}=\frac{1}{\sqrt{n}}\sum_{b\in B} c_b,
\end{equation*}
so
\begin{equation*}
    H=    n i \tilde{c_1}\tilde{c_2}.
\end{equation*}
Since $i\tilde{c_1} \tilde{c_2}$ is Hermitian and satisfies $(i\tilde{c_1} \tilde{c_2})^2=\mathbb{I}$, it has eigenvalues in $\{\pm 1\}$. Thus the eigenvalues of $H_n$ are $\{\pm n\}$.  Note that $H_n$ is $n$-sparse, $m=n^2$, and $\|H_n\| = \lambda_{\max}(H_n)$ so that 
\begin{equation*}
\|H_n\| = \lambda_{\max}(H_n) = n = \Theta\left(\frac{n^2}{2n+1}\right) = \Theta\left(\frac{m}{qk+1}\right).
\end{equation*}



\bibliography{approximating-extremal-eigenvalues-lipics}
\end{document}